\documentclass[letterpaper]{llncs}

\usepackage{amsmath,amssymb,amsfonts}
\usepackage{tgpagella,eulervm}
\usepackage{varioref}
\usepackage{enumerate}

\newcommand{\es}{\varnothing}

\DeclareMathOperator{\mex}{\mathsf{mex}}
\DeclareMathOperator{\nimsum}{\mathsf{nim-sum}}

\begin{document}

\title{{\sc $P_3$-Games on Chordal Bipartite Graphs}}

\author{
Wing-Kai~Hon\inst{1}
\and 
Ton~Kloks\inst{}
\and 
Fu-Hong~Liu\inst{1}
\and 
Hsiang-Hsuan~Liu\inst{1,2}
\and
Tao-Ming~Wang\inst{3}
\and 
Yue-Li~Wang\inst{4}
}

\institute{National Tsing Hua University, Hsinchu, Taiwan\\
{\tt (wkhon,fhliu,hhliu)@cs.nthu.edu.tw}
\and
University of Liverpool, Liverpool, United Kingdom\\
{\tt hhliu@liverpool.ac.uk}
\and
Tunghai University, Taichung, Taiwan\\
{\tt wang@go.thu.edu.tw}
\and
National Taiwan University of Science and Technology\\
{\tt ylwang@cs.ntust.edu.tw}
}

\maketitle 

\begin{abstract}
Let $G=(V,E)$ be a connected graph. 
A set $U \subseteq V$ is convex if $G[U]$ 
is connected and 
all vertices of $V \setminus U$ have at most one neighbor in $U$. 
Let $\sigma(W)$ denote the unique  smallest convex set that contains 
$W \subseteq V$. 

Two players play the following game. Consider a convex set $U$ and call it 
the `playground.' 
Initially, $U = \es$. 
When $U=V$, the player to move loses the game. Otherwise, that 
player chooses a vertex $x \in V \setminus U$ which is 
at distance at most two 
from $U$. The effect of the move is that the playground $U$ changes into 
$\sigma(U \cup \{x\})$ and the opponent is presented with this 
new playground. 

A graph is chordal bipartite if it is bipartite 
and has no induced cycle of 
length more than four. 
In this paper we show that, when $G$ is 
chordal bipartite, there is a polynomial-time algorithm that computes  
the Grundy number of the $P_3$-game played on $G$. 
This implies that there is an efficient algorithm to decide 
whether the first player has a winning strategy. 
\end{abstract}

\section{Introduction}

The $P_3$-convexity in graphs was introduced in~\cite{kn:centeno}. 

\begin{definition}
Let $G$ be a connected graph. 
A set $U$ is \underline{convex} if 
\[ \boxed{\text{\rm $G[U]$ is connected} \quad\text{and}\quad 
\forall_{x\in V \setminus U} \; |N(x) \cap U| \leq 1. }\] 
\end{definition}
Notice that the intersection of any two convex sets is convex. 
Since $\es$ and $V(G)$ are also convex, 
it follows that this convexity is an alignment. 
Let $\mathcal{L}=\mathcal{L}(G)$ denote 
the collection of convex sets in $G$. 

\bigskip 

The \underline{convex closure} of a set 
$W \subseteq V$ is defined as the smallest convex 
set that contains $W$, that is, 
\[ \sigma(W)= \bigcap \;\{\;U\;|\; W \subseteq U \quad \text{and}\quad 
U \in \mathcal{L} \;\}.\] 

\bigskip 

We introduced the \underline{$P_3$-game} in~\cite{kn:hon}. It is played as follows. 
Let $G$ be a connected graph. When it is his turn, each of the two players 
is presented with a playground. The playground is a convex set $U$ in $G$. 
Initially, $U = \es$. When it is a player's turn to move,  
he loses the game if $U=V$. Otherwise, he selects a 
vertex $x \in V\setminus U$ at distance at most two from $U$. The effect 
of the move is that the playground changes into $\sigma(U \cup \{x\})$, and 
it is now the opponent's turn to move. 

\bigskip 

In this paper we show that, when $G$ is chordal bipartite, there exists an efficient  
algorithm that decides 
whether the player who is first to move has a winning strategy. 

\section{Preliminaries}

\subsection{Chordal bipartite graphs}

Golumbic and Goss launched the studies on chordal bipartite graphs. 
They defined the class of graphs as follows. 

\begin{definition}
A graph is \underline{chordal bipartite} if it is bipartite and has 
no induced cycles of length more than 4. 
\end{definition}

Several characterizations of this class of graphs are available. We refer 
to~\cite{kn:kloks2},  which contains  a short survey. 
In this paper also  appears the following 
lemma~\cite[Lemma~2 and Corollary~1]{kn:kloks2}. 

\begin{lemma}
\label{lm cb}
Let $G=(A,B,E)$ be chordal bipartite, let $S$ be a minimal separator in $G$, 
and let 
$C$ be a component that is close to $S$, that is, 
$C$ is a component of $G-S$ and $N(C)=S$. Then 
\begin{enumerate}[\rm (I)]
\item $G[S]$ is complete bipartite (possibly an independent set),
If $S \cap A \neq \es$, then there exists a vertex $x \in C$ satisfying  
\[N(x) \cap S=S \cap A,\]
\item if $S \cap A \neq \es$ and $S \cap B \neq \es$, there exist 
two adjacent vertices $x$ and $y$ satisfying 
\[N(x) \cap S=S \cap A \quad\text{and}\quad N(y) \cap S=S \cap B.\]
\end{enumerate}
\end{lemma}
This implies that each color class of a  minimal separator 
is the common neighborhood of two nonadjacent vertices. 
If the separator has vertices in both color classes, 
then there exists a $2K_2$ such that each color class 
of $S$ is the common neighborhood of one of the two 
nonadjacent pairs, that are in the same color class of $G$,  
and that are in the $2K_2$. 
Golumbic and Goss show that this condition on the separator 
with vertices in both color classes, characterizes chordal 
bipartite graphs; namely, a graph is chordal bipartite 
if and only if 
every minimal edge-separator (separating two edges of $G$) 
is complete bipartite. 

\subsection{Game graphs and Grundy numbers} 

The game graph, for the $P_3$-game defined above on a chordal 
bipartite graph $G$, is a 
directed acyclic graph $P$, whose vertices are the playgrounds. 
There is an arc from a playground $A$ to a playground $B$ if 
$B$ can be reached from $A$ within one move. 

\bigskip 

The game graph $P$ is labeled as follows. 
The unique sink-node, $V$, is labeled $0$. Recursively, let 
$A$ be an unlabeled playground for which all outgoing neighbors 
are labeled. Then the label of $A$ is the $\mex$-value of 
its successors. 

\begin{definition}
\label{df mex}
The \underline{$\mex$-value} of a set of 
nonnegative integers is the 
smallest nonnegative integer that is not in the set.
\end{definition}

\begin{definition}
The \underline{Grundy value} of $G$ is the label of $\es$ 
in the game graph $P$. 
\end{definition}
We denote the Grundy value of the graph $G$ by $g(G)$. 

\bigskip 

The player who is to move first, wins the game 
if and only if the Grundy value is not 0. Thus, the game 
graph $P$ provides an (exponential) algorithm to decide 
the $P_3$-game on a graph $G$. 

\bigskip 

The Sprague-Grundy theorem deals with products of games. 
Let $\mathcal{G}$ be a collection of games. The product game 
of $\mathcal{G}$ is the game in which each player makes a 
(legal) move in 
one of the games of $\mathcal{G}$. The Sprague-Grundy theorem 
is the following. 

\begin{theorem}
\label{thm SG}
Let $\mathcal{G}$ be a collection of impartial 2-person games. 
Then the Grundy value of the product game is the nim-sum of 
the Grundy values of the games in $\mathcal{G}$. 
\end{theorem}

\section{The $P_3$-game on biconnected chordal bipartite graphs}

Centeno, et al., showed that, in a biconnected chordal graph $G$, if 
$x$ and $y$ are two vertices at distance at most two, then 
$\sigma(\{x,y\})=V(G)$.
The following lemma shows that a similar statement holds in biconnected, 
chordal bipartite graphs. 

\begin{lemma}
\label{lm biconnected}
Let $G=(A,B,E)$ be a biconnected and chordal bipartite graph. Let $U$ be a convex 
set that contains two nonadjacent vertices $x$ and $y$ that are in a $C_4$ of $G$. 
Then 
\[\sigma(\{x,y\})=V(G).\]
\end{lemma}
\begin{proof}
Let $S$ be a minimal $x,y$-separator, and let $C_x$ 
and $C_y$ be the components that contain $x$ and $y$. 
Then $S$ contains the common neighbors of $x$ and $y$. 

\medskip 

\noindent
By Lemma~\ref{lm cb}, there exist vertices in $C_x$ and $C_y$ 
that are adjacent to all vertices of $S$ in one color class. It 
follows that there exist vertices in $C_x$ and in $C_y$ that have 
two neighbors in $U$, and so, they are also in $U$. 
In turn, this implies that $S \subseteq U$. 

\medskip 

\noindent
Let $\Omega \subseteq C_x$ be the set of vertices that are adjacent to 
every vertex of one of the two color classes of $S$. 
Let $X_1,\dots,X_t$ be the components of $C_x\setminus \Omega$. 
Then $S_i=N(X_i)$ is a minimal separator, and it is contained in $U$. 
By Lemma~\ref{lm cb}, and by induction, it now follows that $X_i \subset U$. 
This proves that $C_x \subset U$ and, similarly, $C_y \subset U$. 

\medskip 

\noindent
Let $D$ be a component of $G-S$, other than $C_x$ or $C_y$. 
Then $N(D)$ is a minimal separator. Since $G$ is biconnected, 
\[N(D) \subseteq S \quad \text{and}\quad |N(D) \cap S| \geq 2.\] 
By the same argument as above, $D \subset U$. 

\medskip 

\noindent
This proves the lemma.
\qed\end{proof}

\bigskip 

\begin{example}
Consider a ladder $L$. Notice that every convex set is either 
\begin{enumerate}[\rm (a)]
\item $\es$, or $V(L)$, or a single vertex, or 
\item a rung of $L$, or 
\item a connected subpath of a stile (stringer) of $L$.
\end{enumerate}
\end{example}

\bigskip 

\begin{theorem}
\label{thm biconnected}
Assume $G$ is a biconnected, chordal bipartite graph with 
at least two vertices. Then the second player to move has a 
winning strategy. 
\end{theorem}
\begin{proof}
If $G$ is an edge, the second player wins the $P_3$-convex game. 
Assume $G$ has more than 2 vertices. 
Assume the first move labels a vertex $s$. Notice that 
$s$ is in a $C_4$, otherwise $G$ would have a cutvertex. 
The player to move, chooses a vertex $s^{\prime}$ which is not adjacent 
to $s$ and which is in a $C_4$ together with $s$. Then, 
by~Lemma~\ref{lm biconnected}, the second move changes the playground 
into $V(G)$, which ends the game. 
\qed\end{proof}

\begin{remark}
Similarly, when $s$ is a pendant vertex and $G-x$ is a
biconnected and chordal bipartite graph with at least 
two vertices, 
then the player who 
is first to move has a winning strategy. When $G$ is $P_3$, 
then the winning move is to play the midpoint. 
Otherwise, when $G$ has at least 4 vertices, 
the player who moves first labels the pendant vertex $s$. His 
opponent either adds a vertex or an edge of the biconnected 
component to the playground and, since the vertex or edge is 
in a $C_4$, the player who made the first move can then end the game. 
\end{remark}

\section{Splitters}

Let $G$ be a connected chordal bipartite graph. 
A generalization of the $P_3$-game is, where the initial playground is 
some (arbitrary) convex set $U$, instead of $\es$. 
We denote the Grundy number of this game 
by $g^{\ast}(U)$, or by $g^{\ast}(G,U)$, when 
the graph $G$ is not clear from the context. 
Then we have, for the Grundy value $g(G)$ of $G$, 
\[g(G)=g^{\ast}(\es)=g^{\ast}(G,\es).\]

\bigskip 

\begin{definition}
A \underline{splitter} is a playground that contains a minimal 
separator of $G$. 
\end{definition}

Let $S$ be a minimal separator, and assume $S \subseteq U$, for some 
playground $U$, 
and let $C_1,\dots,C_t$ 
be the components of $G-S$. 
Denote 
\[\Bar{C}_i=C_i \cup N(C_i), i \in \{\;1,\;\dots,t\;\}.\] 
Then each player, when it is his move, plays a vertex in 
one of the components $C_i$, that is, he plays a move in one 
of the games $G_i=G[\Bar{C}_i]$ with playground $U_i=U \cap V(G_i)$. 
By Theorem~\ref{thm SG}, this prove the following theorem.

\begin{theorem}
\label{thm splitter}
Let $G$ be a connected chordal bipartite graph. 
Let $U$ be a splitter of the $P_3$-game played on $G$. Then   
\[g^{\ast}(G,U)= \nimsum \;\{\;g^{\ast}(G_i,U_i)\;|\; 
i\in \{\;1,\;\dots,\;t\;\}\;\}.\]
\end{theorem}

\begin{remark}
Notice that the definition of $\Bar{C}_i$ guarantees 
that any move made in the product game is a legal move 
in the $P_3$-game on $G$.  
Notice also the necessity of the condition 
that the separator $S$ is part of the current playground; 
otherwise, a move within $S$ would be a move in all 
games $G_i$, which is not allowed in the product game. 
\end{remark}

\section{Deciding the $P_3$-game on chordal bipartite graphs}

Let $G$ be a connected, chordal bipartite graph. 
Let $P$ be the game graph of the $P_3$-game played on $G$. 
Let $P^{\prime}$ be the labeled digraph obtained from $P$ as follows. 
Let $U$ be a playground. Let $H(U)$ be the graph obtained from $G$ by 
removing those vertices $x \in V(G)$ that satisfy 
\[\sigma(U\cup\{x\})=V.\] 

\begin{definition}
The \underline{augmented game graph} 
$P^{\ast}$ is the labeled digraph obtained from $P^{\prime}$, 
by adding an arc from each sink in $P^{\prime}$ to a new sink node $V$. 
\end{definition}

\subsubsection{The augmented game.}

The augmented game graph $P^{\ast}$ represents the following `augmented' game. 
When it is a player's move, and when the playground is a convex set $U$, 
then he chooses a vertex from $H(U)$ which is at distance at most 2 from $U$. 
When he cannot make a move, he loses the game. 

\bigskip 

\begin{lemma}
Assume that $G$ is biconnected, and chordal bipartite, and 
assume that $G$ has at least 2 vertices.  
The Grundy value of $G$ satisfies 
\[g(G)=g^{\ast}(G),\]
where $g^{\ast}(G)$ is the Grundy value of the augmented game. 
\end{lemma}
\begin{proof}
Each move in the $P_3$-game on $G$ is a move in $P$, 
\[A \rightarrow B,\]
where $A$ and $B$ are convex sets. That is, either it is a move 
\[A \rightarrow V,\] 
or else it is a move in $P^{\prime}$. 

\medskip 

\noindent
It follows that for each convex set $U$ in $P$, the Grundy value is 
\begin{equation}
\label{eq1}
g(U)= 
\begin{cases}
\mex \;\{\; 0,\; g^{\ast}(U) \;\}, & \text{if $U \rightarrow V$ in $P$,}\\
g^{\ast}(U) & \text{otherwise,}
\end{cases}
\end{equation}
where $g^{\ast}(U)$ is the Grundy value of $U$ in the augmented game. 
By Lemma~\vref{lm biconnected}, since $G$ is biconnected, when $U \neq V$
and $U \neq \es$, there 
exists a vertex $x \in V \setminus U$ satisfying 
\[\sigma(U \cup \{x\})=V.\] 
This implies that, unless $U=V$ or $U=\es$, there is an 
arc $U \rightarrow V$, and so,  
\[g^{\ast}(U)= \mex \;\{0,\;g^{\ast}(U)\;\}.\]

\medskip 

\noindent
Notice that $g(\es)=g^{\ast}(\es)$, since no vertex played as 
an initial move ends the game (since $G$ has at least two vertices). 
Finally, by definition, 
\[g(V)=g^{\ast}(V)=0.\]

\medskip 

\noindent
This proves the lemma. 
\qed\end{proof}

\bigskip 

\begin{theorem}
There exists a polynomial-time algorithm to compute the Grundy value 
of the $P_3$-game on chordal bipartite graphs. 
\end{theorem}
\begin{proof}
We may assume that $G$ is connected. 
Consider a playground $U$ that contains an induced 
$P_3$ in $G$, say $[x,y,z]$. 
We claim that $U$ is a splitter. We may assume that none of $x$, $y$ 
or $z$ is a cutvertex in $G$, otherwise we are done. Thus $\{x,y,z\}$ 
is contained in a biconnected component of $G$. We may assume also 
that $U \neq V$. 

\medskip 

\noindent
Let $C_y$ be the component of $G-N[x]$ that contains $y$ and let
\[S=N(C_y).\] 
Then $S \subseteq N(x)$, and so $S$ is an independent set (since $G$ is bipartite). 
By Lemma~\vref{lm cb}, there exists a nonempty set of vertices 
$Y^{\prime} \subseteq C_y$ which are adjacent to all vertices of $S$. 
We claim that $Y^{\prime} \cap H(U) =\es$, where $H(U)$ is  
the label of $U$ in the augmented 
game graph $P^{\ast}$. 

\medskip 

\noindent
To see this, first notice that 
$y^{\prime} \notin U$, for $y^{\prime} \in Y^{\prime}$. Otherwise, since $U$ is contained in a 
biconnected component of $G$, there exists a vertex $q \in S \setminus \{z\}$, 
which is a common neighbor of $x$ and $y^{\prime}$. This implies that 
$\{x,z,y^{\prime},q\}$ induces a $C_4$. By Lemma~\vref{lm biconnected},
$y^{\prime} \in U$ implies that $\sigma(U)=U=V$, which is a contradiction. 

\medskip 

\noindent
Every vertex of $Y^{\prime}$ represents a 
\underline{legal} move, since they are adjacent 
to $z \in U$. 
Since $Y^{\prime}$ contains legal moves $y^{\prime}$ for which 
\[\sigma(U\cup \{y^{\prime}\})=V,\] 
we have, by definition of the augmented game graph, 
$Y^{\prime} \cap H(U)=\es$. 

\medskip 

\noindent
In the augmented game, the vertices of $Y^{\prime}$ are removed from the 
graph (recall that these represent moves that point to $V$). 
Let $G^{\prime}=G-Y^{\prime}$. 
Let $\{D_i\}$ represent the set of components of $G[C_y]-Y^{\prime}$. 
First consider a component $D_i$ that contains a vertex of $U$. 
Let $S_i=N(D_i)$. Then $S_i \subseteq S$. As long as $|S_i|>0$, by the 
argument above, we find new vertices in $D_i$ that are not in $H(U)$.  
Thus, after removal of all vertices that are not in $H(U)$, we have that 
all components $D_i$ that contain a vertex of $U$, satisfy $N(D_i)=\{z\}$.%
\footnote{Notice that the same argument applies to components of $G-N[x]$ 
that have no vertices of $U$, but are adjacent to $z$. After removal of 
vertices that are not in $H(U)$,  these components split also off  
as a component of $G-\{z\}$.} 
This proves that $H(U)$ has a cutvertex, and so, in the augmented game graph,  
$U$ is a splitter. 

\medskip 

\noindent
Notice that there are at most $n^3$ minimal splitters, 
since it is bounded from above by the number of  induced $P_3$s in $G$. 
Each connected subgraph of $P^{\ast}$, without 
splitters has at most $O(n^2)$ nodes, since otherwise it contains a splitter. 
Therefore, the number of nodes in the augmented game is $O(n^5)$. 
It follows that the computation of the augmented game graph 
can be carried out in 
polynomial time. The Grundy values can be computed using 
Theorem~\ref{thm splitter} (ie, the $\nimsum$-operator, and~\eqref{eq1} (which 
relates the Grundy values in $P$ and $P^{\ast}$), and the $\mex$-operator. 

\medskip 

\noindent
This proves the theorem. 
\qed\end{proof}

\section{Concluding remark}

In this paper we introduced a new technique, dubbed `splitters,' 
which is used for the computation of the Grundy numbers of certain 
games on graphs. The technique attempts to reduce the game graph, 
by using splitters, 
to an equivalent game graph which has polynomial size. In the case of 
the $P_3$-game on chordal bipartite graphs, this turned out to be 
successful. 
At the moment we are investigating for which classes of graphs, and for 
which games, this technique is applicable.  It would be nice to have a 
characterization of the classes of graphs, say with a polynomial number 
of separators, for which the $P_3$-game is solvable in polynomial time. 

\section{Acknowledgement}

Ton~Kloks thanks the Department of Computer Science at National Tsing Hua 
University for their kind hospitality and support.

\end{document}